\documentclass[11pt]{article}

\usepackage[T1]{fontenc}
\usepackage[utf8]{inputenc}
\usepackage{csquotes}

\usepackage[margin=1in]{geometry}
\usepackage{ifthen}
\usepackage{graphicx,color,xcolor}
\usepackage[authoryear]{natbib}
\definecolor{cornellred}{rgb}{0.7, 0.11, 0.11}
\definecolor{dgreen}{rgb}{0.0, 0.5, 0.0}
\definecolor{ballblue}{rgb}{0.13, 0.67, 0.8}
\definecolor{royalblue(web)}{rgb}{0.25, 0.41, 0.88}
\definecolor{bleudefrance}{rgb}{0.19, 0.55, 0.91}
\definecolor{royalazure}{rgb}{0.0, 0.22, 0.66}
\usepackage{url}
\usepackage{hyperref}
\hypersetup{
    colorlinks = true,
    linkcolor=cornellred,
    citecolor=royalazure,
    linkbordercolor = white,
    urlcolor  = cornellred
}

\usepackage{fancybox}

\usepackage{microtype}
\usepackage{graphicx}
\usepackage{subfigure}
\usepackage{booktabs} 
\usepackage{hyperref}
\usepackage{xspace}
\usepackage{xcolor}
\usepackage{amsmath,amssymb,amsthm}
\usepackage{algorithm,algorithmic}
\usepackage{pgf,tikz}
\usepackage{tkz-euclide}
\usetkzobj{all}
\usetikzlibrary{arrows}
\usetikzlibrary{chains,arrows}
\usepackage{dsfont}
\usepackage{todonotes}
\usepackage{mathtools}
\usepackage{color-edits}
\addauthor{rn}{red}     
\addauthor{vc}{blue}    
\addauthor{ms}{green}   
\theoremstyle{plain}
\usepackage{romannum}
\usepackage{cleveref}
\newtheorem{theorem}{Theorem}[section]
\newtheorem{lemma}[theorem]{Lemma}

\newtheorem{corollary}[theorem]{Corollary}

\newtheorem{definition}{Definition}[section]
\newtheorem{remark}[theorem]{Remark}
\usepackage{bbm}

\newcommand{\ie}{{i.e.,\ }}
\newcommand{\eg}{{e.g.,\ }}

\newcommand{\eps}{\epsilon}

\definecolor{ttttff}{rgb}{0.2,0.2,1.}
\definecolor{ttqqqq}{rgb}{0.2,0.,0.}
\definecolor{ududff}{rgb}{0.30196078431372547,0.30196078431372547,1.}
\definecolor{ccqqqq}{rgb}{0.8,0.,0.}

\newcommand{\dfun}{g}
\newcommand{\vfun}{f}
\newcommand{\mbench}{\texttt{MAX-upper}}

\newcommand{\Dmbench}{\texttt{1D-MAX-upper}}
\newcommand{\Dsbench}{\texttt{1D-SUM-upper}}
\newcommand{\ind}[1]{\mathbbm{1}\left[ #1 \right]}
\newcommand{\Ev}{\mathcal{E}}

\begin{document}
\pagenumbering{arabic}
\onecolumn



\title{Hierarchical Clustering for Euclidean Data}
\author{Moses Charikar\thanks{Stanford University.}\\{\tt moses@cs.stanford.edu}\and Vaggos Chatziafratis\footnotemark[1]\\{\tt vaggos@cs.stanford.edu} \and Rad Niazadeh\footnotemark[1]\\{\tt rad@cs.stanford.edu} \and Grigory Yaroslavtsev\thanks{Indiana University, Bloomington. Research supported by NSF award CCF-1657477 and Facebook Faculty Research Award.}\\{\tt grigory@grigory.us}}
\date{\today}
\maketitle

\begin{abstract}
Recent works on Hierarchical Clustering (HC), a well-studied problem in exploratory data analysis, have focused on optimizing various objective functions for this problem under \emph{arbitrary} similarity measures. In this paper we take the first step and give novel scalable algorithms for this problem tailored to Euclidean data in $\mathbb{R}^d$ and under vector-based similarity measures, a prevalent model in several typical machine learning applications. We focus primarily on the popular Gaussian kernel and other related measures, presenting our results through the lens of the objective introduced recently by~\cite{joshNIPS}. We show that the approximation factor in \cite{joshNIPS} can be improved for Euclidean data. We further demonstrate both theoretically and experimentally that our algorithms scale to very high dimension $d$, while outperforming average-linkage and showing competitive results against other less scalable approaches.
\end{abstract}

\section{Introduction}\label{sec:intro}
\newcommand{\Dfun}{\mathcal{F}^{-}}
\newcommand{\Jfun}{\mathcal{F}^{+}}




Hierarchical clustering is a popular data analysis method, with various applications in data mining~\citep{berkhin2006survey}, phylogeny~\citep{eisen1998cluster}, and even finance~\citep{tumminello2010correlation}. In practice, simple agglomerative procedures like average-linkage, single-linkage and complete-linkage are often used for this task
(See the book by \citet*{MRS08} for a comprehensive discussion). While the study of hierarchical clustering has focused on algorithms, there was a lack of objective functions to measure the quality of the output and compare the performance of algorithms.
To remedy this, \cite{dasguptaSTOC} recently introduced and studied an interesting objective function for hierarchical clustering for similarity measures.

The input to the above problem is a set of points $V$ with a similarity measure $w_{ij}$ between $i$ and $j$. 
Given a hierarchical clustering represented by a tree $\mathcal T$ whose leaves correspond to points in $V$, Dasgupta's objective is defined as 
$\Dfun(\mathcal T) = \sum_{(i,j) \in E} w_{ij} |\{x \in T(i,j)\}|,$
where $T(i,j)$ is the subtree of $\mathcal T$ rooted in the least common ancestor of $i$ and $j$ in $\mathcal T$.  \citeauthor{dasguptaSTOC} showed that solutions obtained from minimizing this objective has several desirable properties, which prompted a line of work on objective driven algorithms for hierarchical clustering, resulting in new algorithms as well as shedding light on the performance of classical methods~\citep{royNIPS,vaggosSODA,vincentSODA}.


%
Recently, this 
viewpoint has been applied to understand the performance of average linkage agglomerative clustering, one of the most popular algorithms used in practice. 
\cite{joshNIPS} introduced a new objective, in some sense dual to the objective introduced by Dasgupta:
  \begin{equation*}
\label{eq:obj-2}
\Jfun(\mathcal T) = \sum_{(i,j) \in E} w_{ij} (|V| - |\{x \in T(i,j)\}|),
\end{equation*}
They showed that Average-Linkage obtains a $\frac{1}{3}$-approximation for maximizing this objective function.

It turns out that a random solution 
also achieves a $\frac{1}{3}$-approximation ratio for this problem.
Recently \citet*{charikar2018hierarchical} showed that in the worst case, the approximation ratio achieved by Average-Linkage is no better than $\frac{1}{3}$ for maximizing $\Jfun(\mathcal T)$.
They also gave an SDP based algorithm that achieves a ($\frac{1}{3}+\epsilon$)-approximation for the problem, for small constant $\epsilon$.


One drawback of the prior work on these hierarchical clustering objectives is the fact that they all consider arbitrary similarity scores (specified as an $n \times n$ matrix);
however, there is much more structure to such similarity scores in practice.
In this paper, we initiate the study of the commonly encountered case of \emph{Euclidean data}, where the similarity score $w_{ij}$ is computed by applying a monotone decreasing function to the Euclidean distance between $i$ and $j$. Roughly speaking, we show how to exploit this structure to design improved approximation/faster algorithms for hierarchical clustering, and how to re-analyze algorithms commonly used in practice.
In this paper, we focus on the problem of maximizing $\Jfun(\mathcal T)$, introduced by \cite{joshNIPS}.


Arguably the most common distance-based similarity measure used for Euclidean data is the Gaussian kernel. Here we use the spherical version $w_{ij} \sim \exp(-\|v_i - v_j\|/2\sigma^2)$.  The parameter $\sigma^2$, referred to as the \textit{bandwidth}, plays an important role for applications and a large body of literature exists on selection of this parameter~\citep{zelnik2005self}.


As discussed above, devising simple practical algorithms that improve on the $\frac{1}{3}$-approximation for general similarity measures appears to be a major challenge as observed in~\cite{joshNIPS,charikar2018hierarchical}. 
In this paper, \emph{we show that this approximation can be improved for Euclidean data, through fast algorithms that can be scaled to very large datasets.}
While it might seem that the case of Euclidean data with the Gaussian kernel is a very restricted class of inputs to the HC problem, we show that for suitably high dimensions and suitable small choice of bandwidth $\sigma^2$ it can simulate arbitrary similarity scores (scaled appropriately).
Thus any improvements to approximation guarantees for the Euclidean case (that do not apply to general similar scores) must necessarily involve assumptions about the dimension of the data (not very high) or on $\sigma^2$ (not  pathologically small).
Such assumptions on $\sigma^2$  are consistent with common methods for computing the bandwidth parameter based on data (e.g., \citep{zelnik2005self}). 

\paragraph{Contributions.}
We start with the simplest case of 1-dimensional Euclidean data. 
Even in this seemingly simple setting, obtaining an efficient algorithm that produces an exactly optimum solution seems non-trivial, motivating the study of approximation algorithms.
In Section~\ref{sec:1d} we prove that two algorithms -- {\sc Random Cut (RC)} and Average-Linkage (AL) -- obtain $\frac{1}{2}$-approximation of the optimal solution, i.e., tree maximizing $\Jfun(\mathcal T)$ (AL achieves this deterministically, while RC only in expectation). This beats the best known approximation for the general case, which is $0.336$~\citep{charikar2018hierarchical}. Here RC is substantially faster than AL, with a running time of  $O(n \log n)$ vs. $O(n^2)$.

We next consider the high-dimensional case with the Gaussian Kernel and show that Average-Linkage cannot beat the factor $\frac{1}{3}$ even in poly-logarithmic dimensions. 
We propose the \textsc{Projected Random Cut (PRC)} algorithm that gets a constant improvement over $\frac13$, irrespective of the dimension $d$ (the improvement is a function of the ratio of the diameter to $\sigma^2$, and drops as this ratio gets large). 

Furthermore, PRC runs in $O(n(d+\log n))$ time while Average-Linkage runs in $O(dn^2)$ time. Even single-linkage runs in almost-linear time only for constant $d$ and has exponential dependence on the dimension (see e.g.,~\cite{YV18}) and it is open whether it can be scaled to large $d$ when $n$ is large.

\paragraph{Experiments.} Many existing algorithms have time efficiency shortcomings, and none of them can be used for really large datasets. On the contrary, our \textsc{Projected Random Cut} (see \Cref{sec:hd}) is a fast HC algorithm that scales to the largest ML datasets. The running time scales almost linearly and the algorithm can be implemented in one pass without needing to store similarities, so the memory is $O(n)$. We also evaluate its quality on a small dataset (Zoo).

\paragraph{Further related work.}
HC has been extensively studied across different domains (we refer the reader to~\cite{berkhin2006survey} for a survey). One killer application in biology is in phylogenetics~\cite{sneath1962numerical,jardine1968model} and actually popular algorithms, like Average-Linkage are originating in this field. HC is also widely used to perform community detection in social networks~\citep{leskovec2014mining,mann2008use}, is used in bioinformatics~\citep{diez2015novel} and even in finance~\citep{tumminello2010correlation}. Other important applications include image and text classification~\citep{steinbach2000comparison}.

After a formal optimization HC objective was introduced, many works studied HC from an approximation algorithms perspective. The original top-down Sparsest Cut algorithm by Dasgupta, was later shown that it achieves an  $O(\log n)$ approximation by~\cite{royNIPS}; this was finally improved to $O(\sqrt{\log n})$~by \cite{vaggosSODA,vincentSODA}. Subsequent work on beyond-worst-case analysis~\citep{vincentNIPS} proposes a hierarchical stochastic block model under which an older spectral algorithm of~\cite{mcsherry2001spectral} combined with linkage steps gets a $O(1)$-approximation. Finally, alongside the objective of scaling HC to large graphs, the work of~\cite{bateni2017affinity} presents two MapReduce implementations with a focus on minimum-spanning-tree based clusterings that scales to graphs with trillions of edges.

\paragraph{Open problems.}
Here is a list of open questions that we remain unanswered in this paper:
(1) Can one get an improvement over $\frac{1}{3}$ for the problem of maximizing $\Jfun(\mathcal T)$ as a function of $d$ for small $d$?
(2) Can projection on low-dimensional subspaces be used to improve the approximation ratio for the high-dimensional case further?
(3) Does Average-Linkage achieve a $3/4$-approximation for 1-dimensional data?

\section{Preliminaries}\label{sec:prelim}
\paragraph{Euclidean data.} We consider data sets represented as sets of $d$-dimensional \emph{feature} vectors. Suppose these vectors are $v_1,\ldots,v_n\in \mathbb{R}^d$. We focus on similarity measures between pairs of data points, denoted by $[w_{ij}]_{i,j\in[n]}$, where the similarities only depend on the underlying vectors, \ie $w_{ij}=\vfun(v_i,v_j)$ for some function $\vfun:\mathbb{R}^d\times\mathbb{R}^d\rightarrow [0,1]$ and furthermore are fully determined by monotone functions of distances between them.
\begin{definition}[Distance-based similarity measure]
\label{def:dist-base}
	A similarity measure $w_{ij} = \vfun(v_i, v_j)$ is ``distance-based'' if $\vfun(v_i, v_j) = \dfun(\|v_i - v_j\|_2)$ for some function $\dfun \colon \mathbb R \to [0,1]$, and is ``monotone distance-based'' if furthermore $\dfun \colon \mathbb R \to [0,1]$ is a monotone non-increasing function.
\end{definition}

As an example of the monotone similarity measure it is natural to consider the \emph{Gaussian kernel similarity}, \ie
\begin{equation}
w_{ij} = (\sqrt{2 \pi} \sigma)^{-n} e^{- \frac{\|v_i - v_j\|_2^2}{2\sigma^2}},\tag{\emph{Gaussian Kernel}}
\end{equation}
where $\sigma$ is a normalization factor determining the \emph{bandwidth} of the Gaussian kernel~\citep{gartner2003survey}. For simplicity, we ignore the multiplicative factor $(\sqrt{2 \pi} \sigma)^{-n}$ (unless noted otherwise), as our focus is on multiplicative approximations and scaling has no effects. 

\paragraph{Linkage-based hierarchical clustering.} Among various algorithms popular in practice for HC, we focus on two very popular ones: \emph{single-linkage} and \emph{average-linkage}, which are two simple agglomerative clustering algorithms that recursively merge pairs of nodes (or super nodes) to create the final binary tree. Single-linkage picks the pair of super nodes with the maximum similarity weight between their data points, \ie merges $A$ and $B$ maximizing $\max_{i\in A,j\in B}w_{ij}$.  On the contrary, average-linkage picks the pair of super-nodes with maximum average similarity weight at each step, \ie merges  $A$ and $B$ maximizing $\frac{w_{AB}}{\lvert A\rvert\cdot\lvert B \rvert}$, where $w_{AB}:=\sum_{i\in A, j\in B}w_{ij}$. Average-linkage has an approximation ratio of $1/3$ for maximizing the $\Jfun$ objective function~\citep{joshNIPS}, and this factor is tight~\citep{charikar2018hierarchical}.

\paragraph{General upper bounds.} In order to analyze the linkage-based clustering algorithms and our proposed algorithms, we propose a natural upper bound on the value of the $\Jfun$ objective function. The idea is to decompose the objective function $\Jfun$ into contributions of \emph{triple} of vertices: 
\begin{align}
\Jfun(\mathcal T)=\sum_{i<j<k}\left(w_{ij}\ind{\Ev^k_{ij}}+w_{jk}\ind{\Ev^i_{jk}}+w_{ik}\ind{\Ev^j_{ik}}\right)\nonumber
\end{align}
where $\Ev^{x}_{yz}$ denotes the event that $x$ is separated first among the vertices of triple $\{x,y,z\}$ in tree $\mathcal{T}$. Note that the final tree scores only one of the similarity weights between the triple $\{i,j,k\}$. Given this observation, we define the following benchmark:
\begin{align*}
\label{eq:max-bench}
\mbench\triangleq\sum_{i<j<k}\max(w_{ij},w_{jk},w_{ik}),
\end{align*}
Clearly, for all trees $\mathcal{T}$,  $\Jfun(\mathcal T)\leq \mbench$.

\paragraph{One-dimensional benchmarks.} Consider the special case of 1D data points (with any monotone distance-based similarity measure as in \Cref{def:dist-base}), and suppose $v_1\leq v_2\leq \ldots\leq v_n\in \mathbb{R}$. Now, for any triple $i<j<k$,  as a simple observation we have $w_{ik}\leq \min(w_{ij},w_{jk})$. Hence we can modify the above benchmark to obtain two refined new benchmarks:
\begin{align*}
 &\Dmbench\triangleq\sum_{i<j<k}\max(w_{ij},w_{jk}),\\
&\Dsbench\triangleq\sum_{i<j<k}\left(w_{ij}+w_{jk}\right).
\end{align*}
Again, clearly for all trees $\mathcal{T}$ we have: 
\begin{equation*}
\hspace{-1mm}\Jfun(\mathcal T)\leq \Dmbench\leq \Dsbench
\end{equation*}


\section{Hierarchical Clustering in 1D}\label{sec:1d}
\newcommand{\ang}[2]{\angle\left(#1,#2\right)}
\newcommand{\pot}{\Phi}
\newcommand{\als}{\texttt{Score-AL}}
In this section we look at the extreme case where the feature vectors have $d=1$, and we try to analyze the popular/natural algorithms existing in this domain by evaluating how well they approximate the objective function $\Jfun$. We focus on average-linkage and \textsc{Random Cut} (will be formally defined later). In particular, \textsc{Random Cut} is a building block of our algorithm for high-dimensional data given in Section~\ref{sec:hd}.

We show (\romannum{1}) average-linkage gives a $1/2$-approximation, and can obtain no better than $3/4$ fraction of the optimal objective value in the worst-case. We then show (\romannum{2}) \textsc{Random Cut} also is a $1/2$-approximation (in expectation) and this factor is tight.  
In the Appendix~\ref{app:slc-greedy}, we discuss other simple algorithms: single-linkage and greedy cutting (will be defined formally later). We start by the observation that greedy cutting and single-linkage output the same tree (and so are equivalent). We finish by showing (\romannum{3}) there is an instance where single-linkage attains only $\frac{1}{2}$ of the optimum objective value. We further show on this instance both average-linkage and random cutting are almost optimal.

For the rest of this section, suppose we have 1D points $x_1\leq\ldots,\leq x_n\in\mathbb{R}$, where $w_{ij}=\dfun(\lVert x_i-x_j \rVert)$ for some monotone non-increasing function $\dfun:\mathbb{R}\rightarrow [0,1]$.

\paragraph{Average-linkage.} 
In order to simplify the notation, we define $w_{AB}$ to be $\sum_{i\in A,j\in B}w_{ij}$ for any two sets $A$ and $B$. We first prove a simple structural property of average-linkage in 1D.
\begin{lemma}
\label{lemma:AL-neighbour}
For $d=1$, under any monotone distance-based similarity measure $w_{ij} = \dfun(\lVert x_i-x_j \rVert)$,  average-linkage can always merge neighbouring clusters. 
\end{lemma}
\begin{proof}
We do the proof by induction. This property holds at the first step of average-linkage (base of induction). Now suppose up to a particular step of average-linkage, the algorithm has always merged neighbours. At this step, we have an ordered collection of super nodes $(C_1,\ldots,C_m)$, where for every $i<j$ and for all $(x,y)\in C_i\times C_j$, $x\leq y$. If at this step average-linkage doesn't merge two neighbouring super nodes, then there exists $i<j<k$, where 
\[
\frac{w_{C_iC_j}}{\lvert C_i\rvert \lvert C_j\rvert} <\frac{w_{C_iC_k}}{\lvert C_i\rvert \lvert C_k\rvert} \Rightarrow \frac{\sum_{x\in C_i}w_{xC_j}}{\lvert C_j\rvert}<\frac{\sum_{x\in C_i}w_{xC_k}}{ \lvert C_k\rvert}
\]
But note that because of the monotone non-increasing distance-based similarity weights, for any triple $(x,y,z)\in C_i\times C_j\times C_k$ we have $w_{xy}\geq w_{xz}$, This is a contradiction to the above inequality, which finishes the inductive step.
\end{proof}

\begin{theorem}\label{thm:average-linkage}
	For $d = 1$, under any monotone distance-based similarity measure $w_{ij} = \dfun(\lVert x_i-x_j \rVert)$, average-linkage obtains at least $\tfrac{1}{2}$ of the $\Dsbench$, and hence is a $\frac12$-approximation for the objective $\Jfun$.
\end{theorem}
\vspace*{-0.1in}
\begin{proof} The proof uses a potential function argument. Given a partitioning of points $x_1\leq x_2\leq\ldots\leq x_n$ into sets $S_1,\ldots,S_m$, a triple of points $i<j<k$ is called \emph{separated} if no pair of these three points belong to the same set. 
Now, the potential function $\pot$ gets $\{S_1,\ldots,S_m\}$ as input, and maps it to a summation over all separated triples by $\{S_i\}_{i\in[m]}$ as below:
\vspace*{-0.05in}
\begin{equation*}
\pot(S_1,\ldots,S_m)=\sum_{i<j<k:~\textrm{$(i,j,k)$ is separated}}{\left(w_{ij}+w_{jk}\right)}
\end{equation*}
Note that $\pot(\{x_1\},\ldots,\{x_n\})=\Dsbench$, and $\pot(\{x_1,\ldots,x_n\})=0$. 

We now run average-linkage. Based on the definition of $\Jfun$, every time that average-linkage merges two super nodes $A$ and $B$ it scores $w_{AB}\cdot (n-\lvert A\rvert -\lvert B\rvert)$, and sum over of all these per-step scores is equal to its final objective value. Let $\als$ denotes the variable that stores of the score of average-linkage over time. At every step we keep track of (1) the change in the potential function, denoted by $\Delta\pot$ and (2) how much progress average-linkage had towards the final objective value, denoted by $\Delta \als$. In order to prove $1/2$-approximation, it suffices to show that at every step of average-linkage we have:
\begin{equation}
\label{eq:progress}
\Delta\als+\frac{1}{2} \Delta \pot \geq 0. \tag{$\star$}
\end{equation}
To see this note that average-linkage starts with all points separated, and ends with one cluster/super node with all the points. Therefore, by summing \cref{eq:progress} over all merging steps of average-linkage and canceling the terms in the telescopic sum, we have: 
\begin{align*}
&(\Jfun(\mathcal{T}_{\textrm{AL}})-0)+\\
&\frac{1}{2}\left(\pot(\{x_1,\ldots,x_n\})-\pot(\{x_1\},\ldots,\{x_n\})\right)\geq 0.
\end{align*}
Plugging values of $\pot$ at the start and the end, we get:
\[
\Jfun(\mathcal{T}_{\textrm{AL}})\geq \frac{1}{2}(\Dsbench),
\]
which implies the $1/2$-approximation factor. 
	
	To prove \cref{eq:progress}, we focus on a single step of average-linkage where by~\Cref{lemma:AL-neighbour} some two neighboring clusters denoted as $A$ and $B$ get merged. Let $C$ denote the nodes on the left of $A$ and let $D$ denote the nodes on the right of $B$ (\Cref{fig:AL-merge}).\footnote{Note that $C,D$ may be empty sets.} By merging two clusters $A,B$, the change in the score of average-linkage is:
\[
	\Delta\als = w_{AB}(\lvert C\rvert+\lvert D\rvert)
\]	
Moreover, any separated triple $i<j<k$ such that either $i\in C, j\in A,k\in B$ or $i\in A,j\in B,k\in D$ will not be separated anymore after this merge. For each such triple, the potential function drops by $w_{ij}+w_{jk}$. Therefore:
	\[
	-\Delta\pot = \left(w_{AB}|C| + w_{AC}|B|\right)+ \left(w_{AB}|D|+ w_{BD}|A| \right)
	\]
		
To compare the two, we show that 	$w_{AB}(|C|+|D|) \ge w_{AC}|B| + w_{BD}|A|$, and hence:
\begin{align*}
\Delta\als&= w_{AB}(|C|+|D|) \ge w_{AC}|B| + w_{BD}|A|\\
&= -\Delta\pot-\Delta\als~,
\end{align*}
which implies \cref{eq:progress} as desired. To prove the last claim, note that average-linkage picks the pair $(A,B)$ over both $(A,C)$ and $(B,D)$. Therefore, by definition of average-linkage: 
	\[
	\dfrac{w_{AB}}{|A||B|}\ge \dfrac{w_{AC}}{|A||C|}\implies w_{AB}|C| \ge w_{AC}|B|
	\]
	\[
	\dfrac{w_{AB}}{|A||B|}\ge \dfrac{w_{BD}}{|B||D|}\implies w_{AB}|D| \ge w_{BD}|A|
	\]
By summing  above inequalities, we get $w_{AB}(|C|+|D|) \ge w_{AC}|B| + w_{BD}|A|$, which finishes the proof.
\end{proof}
\begin{figure}[ht]
		
		\centering
		\includegraphics[width=8cm]{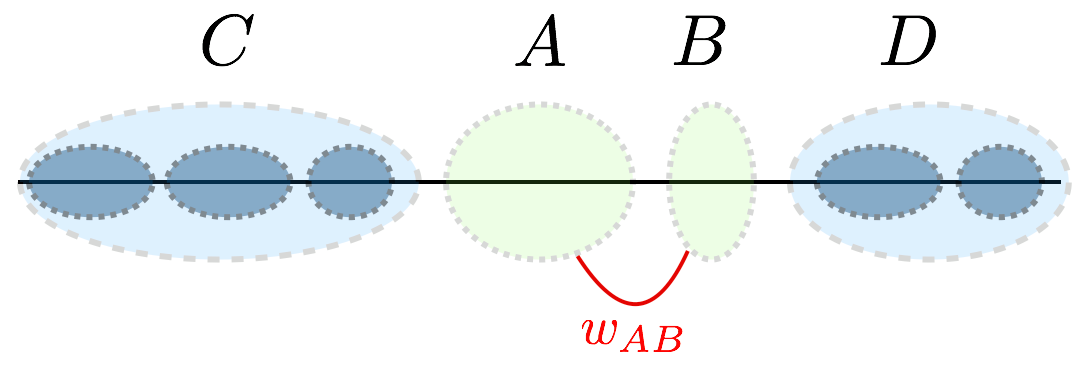}
		\caption{Illustration of the merging process in 1D.}
		\label{fig:AL-merge}
	\end{figure}
As a final note, in \Cref{sec:hard-low-dim} we discuss hard instances for average-linkage under Gaussian kernels when $d=1$, where we essentially show the result of \Cref{thm:average-linkage} is tight when comparing against $\Dsbench$, and there is no hope to get an approximation factor better than $\frac{3}{4}$ for average-linkage in general for $d=1$. 

\paragraph{\textsc{Random Cut}.} The following algorithm (termed as \textsc{Random Cut}) picks a uniformly random point $r$ in the range $[x_1, x_n]$ and divides the set of points into left and right using $r$ as the splitter. The same process is applied recursively until the leaves are reached.

\begin{algorithm}
	\caption{\textsc{Random Cut}}\label{alg:random-cut}
	\begin{algorithmic}
		\STATE {\bfseries Input: } Integer $n$, points $x_1 \le  \dots \le x_n$.\;
		\STATE  {\bfseries Output:} Binary tree with leaves $(x_1, \dots, x_n)$\;
		\STATE\;		
		\IF {$n == 1$} {
			\RETURN New leaf containing $x_1$.
		}
		\ENDIF
		\STATE Pick $r \sim U([x_1, x_n])$ \;
		\STATE Let $m$ be the largest integer such that $x_m \le r$. \;
		\STATE Create new internal tree node $x$. \;
		\STATE $x.left = \textsc{Random Cut}(m, x_1, \dots, x_m)$\;
		\STATE $x.right = \textsc{Random Cut}(n - m, x_{m + 1}, \dots, x_n)$\;
		\RETURN $x$
\end{algorithmic}
\end{algorithm}

\begin{lemma}
For $d = 1$ under any monotone distance-based similarity measure $w_{ij} = g(x_i,x_j)$ the algorithm \textsc{Random Cut} obtains at least $1/2$ fraction of the $\Dmbench$, and  hence gives a $\frac12$-approximation for the objective $\Jfun$ in expectation.
\end{lemma}
\begin{proof}
For every triple $i < j < k$ conditioned on partitioning the interval $[i,k]$ for the first time the longer edge amongst $(x_i, x_j)$ and $(x_j, x_k)$ gets cut with probability $p_1 \ge 1/2$ and the shorter with probability $p_2 \le 1/2$ so that $p_1 + p_2 = 1$.
W.l.o.g let's assume that $(x_i, x_j)$ is the longer edge. 
Then the algorithm \textsc{Random Cut} scores $w_{jk} p_1 + w_{ij}(1 - p_1)$ in expectation for the $(i, j,k)$ triple. 
Note that: 
\begin{align*}
w_{jk} p_1 + w_{ij}(1 - p_1) = (w_{jk}  - w_{ij})\left(p_1 - \frac12\right) + \frac12(w_{ij} + w_{jk})\ge \frac12(w_{ij} + w_{jk})
\end{align*}
By the linearity of expectation taking the sum over all triples $i < j < k$ \textsc{Random Cut} scores at least $1/2$ of $\Dmbench$ in expectation and hence gives $\frac12$-approximation for the objective $\Jfun$.
\end{proof}

\section{Hierarchical Clustering in High Dimensions}\label{sec:hd}
We now describe an algorithm \textsc{Projected Random Cut} which we use for high-dimensional data. This algorithm is given as Algorithm~\ref{alg:gp-random-cut}. It first projects on a random spherical Gaussian vector and then clusters the resulting projections using \textsc{Random Cut}.

\begin{algorithm}
	\caption{\textsc{Projected Random Cut}}\label{alg:gp-random-cut}
	\begin{algorithmic}
		\STATE {\bfseries Input: } Integer $n$, vectors $v_1, \dots, v_n \in \mathbb R^d$.\;
		\STATE  {\bfseries Output:} Binary tree with leaves $(v_1, \dots, v_n)$\;
		\STATE\;		
		\STATE Pick a random Gaussian vector $\mathbf g \sim N_d(0,1)$\;
		\STATE Compute dot products $x_i = \langle v_i, \mathbf g \rangle$ \;
		\STATE $x_{i_1}, \dots x_{i_n} = \textsc{Sort}(x_1, \dots, x_n)$ \;
		\RETURN $\textsc{Random Cut}(n, x_{i_1}, \dots, x_{i_n})$\;		
\end{algorithmic}
\end{algorithm}

\begin{theorem}\label{thm:high-dimensional}
For any input set of vectors $v_1, \dots, v_n \in \mathbb R^d$ the algorithm \textsc{Projected Random Cut} gives an $\alpha $-approximation (in expectation) for the objective $\Jfun$ under the Gaussian kernel similarity measure $w_{ij} \sim e^{-\|v_i - v_j\|_2^2/2\sigma^2}$ where $\alpha = (1 + \delta)/3$ for $\delta = \min_{i,j} \exp( - \frac{\|v_i - v_j\|_2^2}{2\sigma^2})$.
\end{theorem}
\begin{proof}
Recall an upper bound on the optimum: 
$$OPT \le \mbench = \sum_{i < j < k} \max(w_{ij}, w_{ik}, w_{jk}).$$
Fix any triple $(i,j,k)$ where $i < j < k$. Note that the objective value achieved by the algorithm \textsc{Projected Random Cut} can also be expressed as $ALG = \sum_{i < j < k} ALG_{i,j,k}$ where $ALG_{i,j,k}$ is the contribution to the objective from the triple $(i,j,k)$ defined as follows.
Consider the tree constructed by the algorithm. If $v_k$ is the first vector in the triple $(v_i, v_j, v_k)$ to be separated from the other two in the hierarchical partition (starting from the root) then $A_{i,j,k}$ is defined to be $w_{ij}$ (in the other two cases when $i$ or $j$ are separated first the definition is analogous).
Note that since \textsc{Projected Random Cut} is a randomized algorithm $ALG_{i,j,k}$ is a random variable. By the linearity of expectation we have $\mathbb E[ALG] = \sum_{i < j < k} \mathbb E[ALG_{i,j,k}]$.
Thus in order to complete the proof it suffices to show that for every $i < j < k$ it holds that:
$$\mathbb E[ALG_{i,j,k}] \ge \alpha \max(w_{ij}, w_{ik}, w_{jk}).$$

Fix any triple $(v_i, v_j, v_k)$ which forms a triangle in $\mathbb R^d$. Conditioned on cutting this triangle for the first time let $(p_{ij},p_{ik},p_{jk})$ be the vector of probabilities corresponding to the events that the corresponding edge is not cut. I.e. this is the probability that we score the contribution of this edge in the objective. Note that $p_{ij} + p_{ik} + p_{jk} = 1$.

Consider any triangle whose vertices are $v_i, v_j, v_k$. To simplify presentation we set $i = 1, j = 2, k = 3$. We can assume that $\|v_2 - v_1\| \ge \|v_2 - v_3\| \ge \|v_1 - v_3\|$. Let $\theta_1 = \ang{v_1 - v_3}{ v_2 - v_3}$, $\theta_2 = \ang{v_2 - v_1} {v_3 - v_1}$ and $\theta_3 = \ang{v_1 - v_2}{v_3 - v_2}$ so that $\theta_1 \ge \theta_2 \ge \theta_3$. See Figure~\ref{fig:triangle}. Note that the probability that the $i$-th longest side of the triangle has the longest projection is then $\theta_i/\pi$. 

\begin{figure}
\centering
\begin{tikzpicture}[scale=2,rotate=0]
\tkzDefPoint(2,3){A}
\tkzDefShiftPoint[A](0:3){B}
\tkzDefShiftPoint[A](60:1.5){C}
\tkzDefShiftPoint[A](0:2){D}
\tkzDrawSegments(A,B B,C C,A C,D)
\tkzDrawPoints(A,B,C)
\tkzLabelPoint[below left](A){$v_1$}
\tkzLabelPoint[below right](B){$v_2$}
\tkzLabelPoint[above left](C){$v_3$}
\tkzMarkAngle[arc=l, size=8pt](A,C,B)
\tkzLabelAngle[pos=0.4](A,C,B){$\theta_1$}
\tkzMarkAngle[arc=ll, size=8pt](B,A,C)
\tkzLabelAngle[pos=0.5](B,A,C){$\theta_2$}
\tkzMarkAngle[arc=lll, size=8pt](C,B,A)
\tkzLabelAngle[pos=0.5](C,B,A){$\theta_3$}
\tkzDefLine[parallel=through A](C,D) \tkzGetPoint{a}
\tkzDrawLine(A,a)
\tkzLabelLine[pos=0.5,left](A,a){$(\ell_3)$}
\tkzDefLine[parallel=through D](C,D) \tkzGetPoint{d}
\tkzDrawLine(D,d)
\tkzLabelLine[pos=0.24,left](D,d){$(\ell_1)$}
\tkzDefLine[orthogonal=through B](D,C) \tkzGetPoint{p}
\tkzDrawLine(B,p)
\tkzLabelLine[pos=0.7,right](B,p){$(\ell_2)$}
\tkzInterLL(B,p)(D,d) \tkzGetPoint{I}
\tkzMarkRightAngle[size=0.1](B,I,D)
\tkzInterLL(B,p)(A,a) \tkzGetPoint{J}
\tkzMarkRightAngle[size=0.1](B,J,A)
\tkzDrawPoints(D,I,J)
\tkzLabelPoint[below left](D){$x$}
\tkzLabelPoint[right](I){$y$}
\tkzLabelPoint[below left](J){$z$}
\tkzMarkAngle[fill=yellow,opacity=0.5, size=0.8](A,C,D)
\tkzLabelAngle[pos=0.7](A,C,D){$\theta$}
\tkzMarkSegment[color=blue, mark=s|](B,D)
\tkzMarkSegment[color=blue, mark=s||](A,D)
\tkzMarkSegment[color=blue, mark=|](B,I)
\tkzMarkSegment[color=blue, mark=||](J,I)
\tkzDefPointWith[linear, K=0.4](J,I)
\tkzGetPoint{P}
\tkzDefVector[colinear= at J](J,P){T}
\tkzDefVector[orthogonal](J,P){W}
\tkzDrawVector[line width=2pt,color = red](J,T)
\tkzDrawVector[line width=2pt,color = red](J,W)
\tkzLabelPoint[above left](T){$\mathbf g$}
\tkzLabelPoint[above right](W){$\mathbf g^\perp$}
\end{tikzpicture}
\caption{Projecting the triangle $(v_1, v_2, v_3)$ on $\mathbf g$.} \label{fig:triangle}
\end{figure}
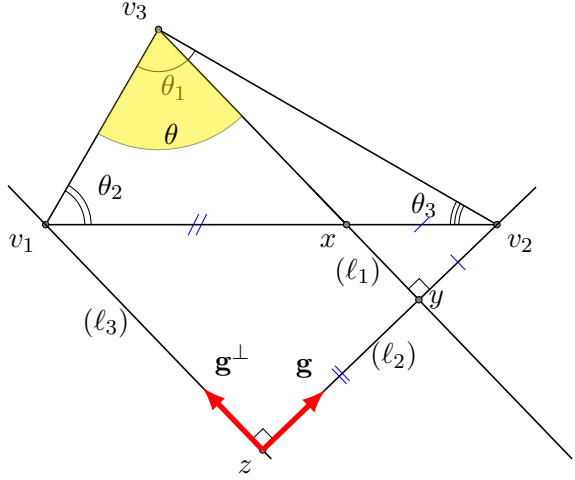
\begin{lemma}
If $(v_1, v_3)$ is the shortest edge in the triangle $(v_1, v_2, v_3)$ then it holds that $p_{13} \ge \frac13$
\end{lemma}
\begin{proof}
Suppose that $\mathbf g$ forms an angle $\pi/2 - \theta$ with $(v_3 - v_1)$, i.e. the vector $\mathbf g^\perp$ orthogonal to $\mathbf g$ forms angle $\theta$ with $v_3 - v_1$. 
We define three auxiliary points $x, y, z$ as follows (see also Figure~\ref{fig:triangle}.). Let $\ell_1$ be a line parallel to $\mathbf g^\perp$ going through $v_3$, let $\ell_2$ to be a line parallel to $\mathbf g$ going through $v_2$ and let $\ell_3$ be the line parallel to $\ell_1$ going through $v_1$.
We then let $x$ be the intersection of $\ell_1$ and $(v_1, v_2)$, $y$ be the intersection of $\ell_1$ and $\ell_2$ and $z$ be the intersection of $\ell_2$ and $\ell_3$ (see Fig~\ref{fig:triangle}).

Thus the projections of $v_3 - v_1$, $v_2 - v_1$ and $v_2 - v_3$ on $\mathbf g$ are $y - z$, $v_2 - z$ and $v_2 - y$.
Hence conditioned on $(v_1, v_2)$ having the longest projection the probability of scoring the contribution of $(v_2, v_3)$ is given as $ p_{23}^{12}(\theta) = \frac{\|y - z\|}{\|y - v_2\|}$ since we are applying the \textsc{Random Cut} algorithm after the projection.
Note that by Thales's theorem we have $p_{23}^{12}(\theta) = \frac{\|y - z\|}{\|y - v_2\|} = \frac{\|x - v_1\|}{\|v_2 - v_1\|}$.
Applying the law of sines to the triangles $(v_1, v_2, v_3)$ and $(x, v_1, v_3)$ we have:
$$\frac{\sin \theta_1}{\|v_1 - v_2\|} = \frac{\sin (\theta_1 + \theta_2)}{\|v_1 - v_3\|}, \quad \frac{\sin \theta} {\|v_1 - v_3\|} = \frac{\sin (\theta + \theta_2)}{\|v_1 - x\|},$$
where we used the fact that $\sin \theta_3 = \sin (\pi - \theta_1 - \theta_2) = \sin (\theta_1 + \theta_2)$. Similarly, we use the fact that $\sin (\ang{v_3 - x}{ v_1 - x}) = \sin (\theta + \theta_2)$.
Using the above we can express $p_{23}^{12}(\theta)$ as:
$$p_{23}^{12}(\theta) = \frac{\sin \theta}{\sin (\theta + \theta_2)} \frac{\sin (\theta_1 + \theta_2)}{\sin \theta_1}$$

Thus the overall probability of scoring the contribution of the edge $(v_1, v_3)$ conditioned on $(v_1, v_2)$ having the longest projection which we denote as $p_{13}^{12}$ is given as:
\begin{align*}
p_{23}^{12} &= \frac{1}{\theta_1}\int_{0}^{\theta_1} p_{23}^{12} d\theta  
\end{align*}

Similarly, consider the probability $p_{13}^{12}$ of scoring the contribution of $(v_1,v_3)$ conditioned on $(v_1,v_2)$ having the longest projection.
We can express it as:
\begin{align*}
p_{13}^{12} &=  \frac{1}{\theta_1} \int_0^{\theta_1} p_{13}^{12}(\theta) d\theta,
\end{align*}
where $$p_{13}^{12}(\theta) = \frac{\sin \theta}{\sin (\theta + \theta_3)} \frac{\sin (\theta_1 + \theta_3)}{\sin \theta_1}.$$ 

Below we will show that $p_{13}^{12} \ge p_{23}^{12}$.
In fact, we will show that for any fixed $\theta \in [0, \theta_1]$ it holds that $p_{13}^{12}(\theta) \ge p_{23}^{12}(\theta)$.
Comparing the expressions for both it suffices to show that $\frac{\sin(\theta_1 + \theta_3)}{\sin (\theta + \theta_3)} \ge \frac{\sin (\theta_1 + \theta_2)}{\sin (\theta + \theta_2)}$ for all $\theta \in [0, \theta_1]$.
Since $\theta_1 = \pi - \theta_2 - \theta_3$ this is equivalent to:
$$\frac{\sin\theta_3}{\sin (\theta + \theta_2)} \le \frac{\sin\theta_2}{\sin(\theta + \theta_3)}$$
It suffices to show that:
$$\sin\theta_3 \sin(\theta + \theta_3) \le \sin \theta_2 \sin (\theta + \theta_2)$$
Using the formula $\sin \alpha \sin \beta = \frac12(\cos(\alpha - \beta) - \cos(\alpha + \beta))$ it suffices to show that:
$$\cos(\theta + 2 \theta_3) \ge \cos(\theta + 2 \theta_2).$$
The above inequality follows for all $\theta \in [0, \pi - \theta_2 - \theta_3]$ since $\theta_3 \le \theta_2$ .
This shows that $p_{13}^{12} \ge p_{23}^{12}$. Since the probability that $(v_1, v_2)$ has the longest projection is $\theta_1 / \pi$ we have that the probability of scoring $(v_1, v_3)$ and $(v_1, v_2)$ having the longest projection is at least $\frac{\theta_1}{2\pi}$.
An analogous argument shows that the probability of scoring $(v_1, v_3)$ and $(v_2, v_3)$ having the longest projection is at least $\frac{\theta_2}{2\pi}$.

Putting things together: 
$$p_{13} \ge \frac12 \frac{\theta_1 + \theta_2}{\pi} = \frac12 \frac{\pi - \theta_3}{\pi}  \ge \frac13,$$
where we used that $\theta_3 \le \pi/3$ since $\theta_3 \le \theta_2 \le \theta_1$.
\end{proof}

We are now ready to complete the proof of Theorem~\ref{thm:high-dimensional}.
Let $\gamma = \frac{\frac23 \delta}{(1 + \delta)}$ and note that $\gamma \ge \delta / 3$ since $\delta \le 1$.
If $p_{13} \ge 1/3 + \gamma$ then the desired guarantee follows immediately.
Otherwise, if $p_{13} \le 1/3 + \gamma$ then we have:
\begin{align*}
\frac{\mathbb E[ALG_{1,2,3}]}{OPT_{1,2,3}} &= \frac{p_{13} w_{13} + p_{12} w_{12} + p_{23}w_{23}}{w_{13}} \\
& \ge \frac13 + \left(\frac23 - \gamma\right)\frac{w_{12}}{w_{13}}  \ge \frac13 + \left(\frac23 - \gamma\right) \delta \\
& = \frac13 + \frac 23 \frac{\delta}{1 + \delta}  \ge \frac{1 + \delta}{3},
\end{align*}
where we used the fact that 
$$\frac{w_{12}}{w_{13}} = e^{\frac{\|v_1 - v_3\|_2^2 - \|v_1 - v_2\|_2^2}{2 \sigma^2}} \ge e^{\frac{- \|v_1 - v_2\|_2^2}{2 \sigma^2}} \ge \delta.$$
\end{proof}

\subsection{Gaussian Kernels with small \texorpdfstring{$\delta$}{} }
  \Cref{thm:high-dimensional} only provides an improved approximation guarantee for \textsc{Projected Random Cut} compared to the factor $1/3$ (\ie the tight approximation guarantees of average-linkage in high dimensions; see \Cref{sec:high-hard}) if $\delta$ is not too small, where $\delta =  \min_{i,j} \exp( - \frac{\|v_i - v_j\|_2^2}{2\sigma^2})$. In particular, we get constant improvement if $\delta =\Omega(1)$. Is this a reasonable assumption? Interestingly, we answer this question in the affirmative by showing that if we have $\delta = \exp(- \tilde \Omega(n))$, then the Gaussian kernel can encode \emph{arbitrary} similarity weights (up to scaling, which has no effects on multiplicative approximations). For simplicity, we only prove this result for $\{\epsilon,1\}$ weights here, while it can be generalized to arbitrary weights.
\begin{theorem}
Given any undirected graph $G=(V,E)$ on $n$ nodes and $\epsilon>0$, there exist unit vectors $\{k_v\}_{v\in V}$ in $\mathbb{R}^d$ and bandwidth parameter $\sigma\in\mathbb{R}_+$, such that  $d=O(n^2)$, $\frac{1}{\sigma^2}=\Omega(n\log(1/\eps))$, and for some $\alpha>0$ we have:
\begin{align*}
&\forall (u,v)\in E: e^{-\frac{\lVert k_u-k_v\rVert_2^2}{2\sigma^2}}=\alpha,\\
&\forall (u,v)\notin E: e^{-\frac{\lVert k_u-k_v\rVert_2^2}{2\sigma^2}}=\alpha\epsilon.
\end{align*}
\end{theorem}
\begin{proof} Our proof is constructive. Let $d=\binom{n}{2}$. Pick orthogonal vectors $\{x_e\}_{e\in E}$ in $\mathbb{R}^d$ such that $\lVert x_e\rVert_2 =\frac{1}{d_u d_v}$, where $d_u$ is the degree of node $u$ in graph $G$. For each $v\in V$, define $y_v\in \mathbb{R}^d$ as follows:
\[
y_v\triangleq\sum_{(u,v) \in E}(d_ud_v)x_{uv}.
\]
Note that $\lVert y_v \rVert_2^2=\sum_{(u, v) \in E}d_u^2d_v^2\frac{1}{d_u^2 d_v^2}=d_v$
As the next step, pick a set of $n$ orthonormal vectors $\{z_v\}$ in the null space of $\{y_v\}_{v\in V}$. Finally, for each $v\in V$, define the final vector $k_v\in\mathbb{R}^d$ as follows:
\[
k_v\triangleq \sqrt{1-\frac{d_v}{n}}z_v+\sqrt{\frac{1}{{n}}}y_v
\]
First, note that these vectors have unit length:
\[
\lVert k_v \rVert_2^2=1-\frac{d_v}{n}+\frac{\lVert y_v\rVert_2^2}{n}=1-\frac{d_v}{n}+\frac{d_v}{n}=1
\]
Now, pick any two vertices $u$ and $v$. If $(u,v)\notin E$, then:
\begin{align*}
\langle k_u,  k_v \rangle &= \langle \sqrt{1-\frac{d_v}{n}}z_u+\sqrt{\frac{1}{{n}}}y_u, \sqrt{1-\frac{d_v}{n}}z_v+\sqrt{\frac{1}{{n}}}y_v\rangle  \\
&=\frac{1}{n} \langle y_u, y_v \rangle=0~,
\end{align*}
where we used the fact that $\langle z_u, z_v\rangle = 0$ , $\langle z_u, y_v\rangle = \langle z_v, y_u\rangle =0$ and the fact that $\langle x_e,  x_{e'}\rangle =0$ for every edge $e$ incident to $u$ and every edge $e'$ incident to $v$, as $e\neq e'$ when $(u,v) \notin E$. Similarly, when $(u,v)\in E$:
\[
\langle k_u, k_v \rangle =\frac{1}{n} \langle y_u, y_v\rangle=\frac{1}{n}(d_u^2d_v^2\lVert x_{uv}\rVert_2^2)=\frac{1}{n}
\]
Now, consider a Gaussian kernel with bandwidth $\sigma=(n\log(1/\epsilon))^{-\frac{1}{2}}$ and vectors $\{k_v\}_{v\in V}$. Since $\|k_u - k_v\|_2^2 = 2(1 - \langle k_u, k_v \rangle)$ from the above calculations of the inner products it follows that:
\begin{align*}
&\forall (u,v)\in E: e^{-\frac{\lVert k_u-k_v\rVert_2^2}{2\sigma^2}}=e^{\frac{1-1/n}{\sigma^2}},\\
&\forall (u,v)\notin E:e^{-\frac{\lVert k_u-k_v\rVert_2^2}{2\sigma^2}}=e^{\frac{1}{\sigma^2}}.
\end{align*}
Thus the ratio is $e^{-\frac{1}{n\sigma^2}}=\epsilon$, as desired.
\end{proof}

\section{Hard Instances for Average-Linkage under Gaussian Kernel Similarity}\label{sec:hard-instances}

\label{sec:high-hard}
\label{thm:high-hard}
\paragraph{High-dimensional case.}
We embed the construction of~\cite{charikar2018hierarchical} shown in Figure~\ref{fig:clique-construction} into vectors with similarities given by the Gaussian kernel.

\begin{theorem}

There exists a set of vectors $v_1, \dots, v_n \in \mathbb R^d$ for $d = poly(\log n)$ for which the average-linkage clustering algorithm achieves an approximation at most $\frac13 + o(1)$ for $\Jfun$ under the Gaussian kernel similarity measure.
\end{theorem}
\begin{proof}[Proof of Theorem~\ref{thm:high-hard}]
We start by this theorem:

\begin{theorem}[\cite{charikar2018hierarchical}]\label{thm:ccn-avg-hardness}
For any constant $\epsilon \in (0,1)$ the instance $\mathcal I$ average-linkage clustering achieves the value of $\Jfun$ at most $\frac16 n^{8/3} + O(n^{7/3})$ while the optimum is at least $\frac12 n^{8/3} - O(n^{7/3})$.
\end{theorem}

\begin{figure}
	\centering
	\usetikzlibrary{shapes.geometric}
	\begin{tikzpicture}[scale=0.9]

	[every node/.style={inner sep=0pt}]
	\foreach \y in {0,2,3} {
	\draw [fill=orange, fill opacity=0.4] (31.25pt, -56.25pt + \y * 2 * 18.75pt) ellipse (120pt and 11pt);	
	\node (\y1) at (-85pt, -46.25pt + \y * 2 * 18.75pt) {\tiny{ $K_{n^{2/3}}$}};	
}
\foreach \x in {0,1,2,6,7,8}{
	\node (\x1) at (-65.5pt +\x*23.44pt, 90pt) {\tiny $K_{n^{1/3}}$};
	\draw [fill=blue, fill opacity=0.4] (-62.5pt +\x*23.44pt,-56.25pt + 3 * 18.75pt) ellipse (11pt and 80pt);	
}
\foreach \x in {0,...,8}{
	\foreach \y in {0,...,3}{
		\node (\x\y)[circle, minimum size=7.5pt, fill=black, line width=0.625pt, draw=black] at (-62.5pt +\x*23.44pt,-56.25pt + \y * 2 * 18.75pt) {};
	}
}
	\draw[thick,<->] (02) -- (03) node[midway,left] {\small$\mathbf{1- \epsilon}$};	
	\draw[thick,<->] (03) -- (13) node[midway,above] {\small$\mathbf 1$};
	\end{tikzpicture}
	\caption{Hard instance $\mathcal I$ from~\cite{charikar2018hierarchical}. Vertices in orange blocks form cliques of size $n^{2/3}$ connected by edges of similarity $1 - \epsilon$, vertices in blue blocks form cliques of size $n^{1/3}$ connected by edges of similarity $1$, all other pairs have similarity $0$.}\label{fig:clique-construction}
\end{figure}
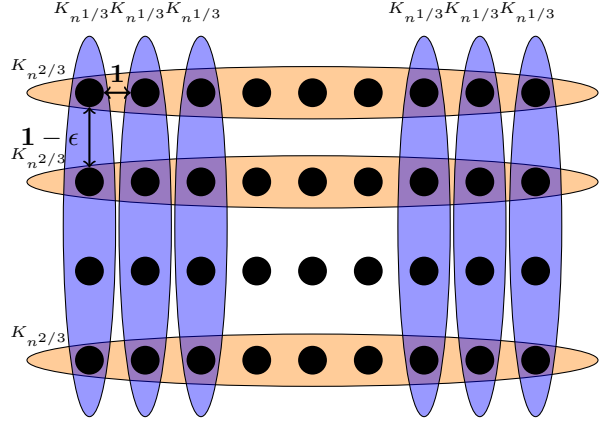

 Let $\Delta, \tau > 0$ be real-valued parameters to be chosen later.
We use indices $i$ and $j$ to index our set of vectors.
For $i \in \{1, 2, \dots, n^{1/3}\}$, $j \in \{1, 2, \dots, n^{2/3}\}$ let
$$v_{i,j} = \Delta (e_i + (1 + \tau)e_{k + j}),$$ 
where $k = n^{1/3}$ and $e_i$ is the $t$-th standard unit vector $e_t = (0 \dots, 1, \dots, 0)$ with the $1$ in the $t$-th entry.
Then it is easy to see that for any fixed $i \in [n^{1/3}]$ and $j_1 \neq j_2 \in [n^{2/3}]$ it holds that:
$$\|v_{i,j_1} - v_{i, j_2}\|_2^2 = 2 (1 + \tau)^2 \Delta^2$$
Similarly, for any fixed $j \in [n^{2/3}]$ and $i_1 \neq i_2 \in [n^{1/3}]$ it holds that:
$$\|v_{i_1, j} - v_{i_2,j}\|_2^2 = 2 \Delta^2.$$
Otherwise if $i_1 \neq i_2 \in [n^{1/3}]$ and $j_1 \neq j_2  \in [n^{2/3}]$ then:
$$\|v_{i_1, j_1} - v_{i_2, j_2}\|_2^2 = 2 \Delta^2 + 2 (1 + \tau)^2 \Delta^2 \ge 4 \Delta^2.$$

By setting $\Delta^2 = 2 \sigma^2 c \log n$ for a sufficiently large constant $c$ the contribution of pairs of vectors with $i_1 \neq i_2$ and $j_1 \neq j_2$ can be made negligible. Let $2 (1 + \tau)^2 \Delta^2 = \alpha$ and $2 \Delta^2 = \beta$.
The rest of the pairs correspond to an hard instance $\mathcal I$ for which average-linkage only achieves a $\frac13 + o(1)$-approximation compared to the optimum.
By setting $\tau = 1 / poly(\log(n))$ we have $e^{(\beta^2 - \alpha^2)/2\sigma^2}  = \Omega(1)$ and hence by Theorem~\ref{thm:ccn-avg-hardness} it follows that average-linkage clustering can't achieve better than $1/3 + o(1)$ approximation for this instance.

Finally, note that by applying the Johnson-Lindenstrauss transform we can reduce the dimension required for the above reduction to $d = poly(\log n)$. Indeed, projecting on a random subspace of dimension $O(\frac{\log n}{\xi^2})$ would preserve $\ell_2^2$-distances between all pairs of vectors up to a multiplicative factor of $(1 \pm \xi)$ setting $\xi = \frac{\tau}{10} = 1/poly(\log n)$ it follows that our hard instance can be embedded in dimension $d = poly(\log n)$.
\end{proof}

%
%

\vspace{-3mm}

\paragraph{Low-dimensional case.}
\label{sec:hard-low-dim}
For $d = 1$ a hard instance for average-linkage clustering can also be constructed.
\begin{lemma}
\label{lemma:four-point}
For points $x_1, \dots, x_n \in \mathbb R$ average-linkage clustering achieves approximation at most $3/4$ for $\Jfun$ under the Gaussian kernel similarity measure.
\end{lemma}
\begin{proof}
Consider the instance consisting of four equally spaced points on a line, i.e. $0, \Delta, 2\Delta, 3 \Delta$.
Then (after shifting the two middle points slightly) the average-linkage clustering algorithm might first connect the two middle points and then connect the two other points in arbitrary order. We denote the cost of this solution as $AVG$.
An alternative solution would be to create two groups $(0, \Delta)$ and $(2\Delta, 3\Delta)$ and then merge them together.
We denote the cost of this solution as $OPT$.
By making $\Delta$ sufficiently large the contribution of pairs at distance more than $\Delta$ from each other can be ignored. We thus have  $AVG \le (\sqrt{2 \pi} \sigma)^{-n} (3 e^{-\Delta^2/2\sigma^2} + o(e^{-\Delta^2/2\sigma^2}))$ and $OPT \ge  (\sqrt{2 \pi} \sigma)^{-n}  4 e^{- \Delta^2/2\sigma^2}$, which gives the ratio of $3/4$ for sufficiently large $\Delta$.
\end{proof}
\begin{corollary}
In the four-point instance of the proof of \Cref{lemma:four-point}, the $\Dsbench$ evaluates to \\
$(\sqrt{2 \pi} \sigma)^{-n} 6 e^{- \Delta^2/2\sigma^2}$, and hence average-linkage cannot obtain more than $1/2$ of $\Dsbench$.
\end{corollary}

\section{Experimental Results}\label{sec:experiments}

\label{sec:experiment}
In this section we demonstrate the quality of the solution returned by \textsc{Projected Random Cut} (\textsc{PRC}) on a small dataset, and highlight that running time only scales linearly on a large dataset.
\textsc{PRC} does not compute the similarity weights (\ie it only takes one pass over the feature vectors with dimension-free memory requirements) and then sorts the projected points in time $O(n\log n)$. Hence, it is fast even in use-cases with over millions of datapoints (in contrast to average-linkage, spectral clustering, or even single-linkage which all have superlinear running times in high dimensions).

We run \textsc{PRC} algorithm on two real datasets from the UCI ML repository \citep{Zoo}: (\romannum{1}) the Zoo dataset \citep{Zoo,dasguptaICML} (the \emph{small} dataset) that contains 100 animals given as 16D feature vectors (this dataset comes from applications in biology), and (\romannum{2}) the SIFT10M dataset \citep{fu2014nokmeans} (the \emph{large} dataset) that contains around 10M datapoints, where each datapoint is a 128D Scale Invariant Feature Transform (SIFT) vector (this dataset comes from applications in computer vision).  For more details, refer to the Appendix~\ref{sec:appendix-experiment}.

\paragraph{Small dataset (Zoo).} We compare \textsc{PRC} to (\romannum{1}) the recursive spectral clustering using the second eigenvector of the normalized Laplacian of the weight matrix (\textsc{Spectral})~\citep{vaggosICML}, (\romannum{2}) average-linkage (\textsc{AL}), and (\romannum{3}) $\mbench$, an upper-bound on the objective value of the optimum tree; see \Cref{sec:prelim}. In contrast to \textsc{PRC}, these three benchmarks need to compute the weights and are slow on large datasets. 

 \Cref{tab:exp} summarizes the results of this experiment for various choices of the parameter $\sigma$ (first column). The second, third and fourth columns report the objective values (\ie $\Jfun$) of the trees returned by $\textsc{PRC}$, $\textsc{Spectral}$ and $\textsc{AL}$ respectively, and the fifth column gives an empirically observed approximation guarantee for PRC by comparing it against the upper bound on optimum $\mbench$. 
 We observe that $\textsc{PRC}$ attains high approximation factors of $\approx [0.73,\ldots,0.92]$ compared to $\mbench$ . Moreover, \textsc{PRC} obtains competitive objective values compared to $\textsc{Spectral}$ and $\textsc{AL}$, although it runs faster (in almost linear-time).

\begin{table}
\small
\begin{center}
 \begin{tabular}{|c||c|c|c|c|c|} 
  \hline
 \small$\sigma$ & \small$\textsc{PRC}$ & \small$\textsc{Spectral}$ & \small$\textsc{AL}$ & \small$\mbench$ & \small$\tfrac{\textsc{PRC}}{\mbench}$ \\ [0.4ex] 
 \hline\hline
 1.5 & 48  & 61 & 28 & 64  & 0.75\\ 
 \hline
 2 & 64  & 83 & 47 & 87  & 0.74\\ 
 \hline
 2.5 & 83  & 100 & 66 & 105  & 0.79\\ 
 \hline
 3 & 100 & 112 & 82 & 117 &0.85\\
 \hline
 3.5 & 111  & 121 & 95& 126 &0.87\\
 \hline
 4 & 117 & 128 & 105 & 132 &0.88\\
 \hline
 4.5 & 123  & 133 & 114 & 137  & 0.91\\ 
 \hline
 5 & 129  & 137 & 120 & 140  & 0.92\\ 
 \hline
\end{tabular}
\end{center}
\caption{Values of the objective (times $10^{-3}$) on the Zoo dataset (averaged over 10 runs).}\label{tab:exp}
\end{table}

On the choice of bandwidth parameter $\sigma$, as a rule of thumb, we find an interval $[\alpha,\beta]$ such that (\romannum{1}) if $\sigma>\beta$, a considerable portion of pairs of datapoints with very different weights under the cosine similarities~\citep{steinbach2000} have almost equal weights under the Gaussian kernel, and (\romannum{2}) if $\sigma<\alpha$,  a considerable portion of pairs of datapoints with almost equal weights under the cosine similarities have very different weights under the Gaussian kernel. We end up with $\alpha=1.5,\beta=5$.

\paragraph{Large dataset (SIFT10M)} The focus of this experiment is measuring the running time of \textsc{PRC}, and showing that it only scales linearly with the dataset size. Note that evaluating the performance of any other algorithm or upper bound (or even one pass over the similarity matrix) would be prohibitive.

We run $\textsc{PRC}$ on truncated versions of SIFT10M of sizes 10K, 100K, 500K, 1M and 10M. $\sigma$ is set to 450~\footnote{Note that the $\sigma$ value does not affect the running time.}. We emphasize that our \textsc{PRC} algorithm runs extremely fast on a 2014 MacBook\footnote{System specs: 8 GB 1600 MHz DDR3 RAM, 2,5 GHz Intel Core i5 CPU.}. The running times are summarized in \Cref{tab:times}. Observe that $\textsc{PRC}$ scales almost linearly with the data and has almost the same running time as just a single pass over the datapoints.

\begin{table}
\small
\begin{center}
 \begin{tabular}{||c||c||c||} 
  \hline
 \small Size & \small \textsc{PRC} (seconds) & \small 1 Data Pass (seconds)  \\ [0.5ex] 
 \hline\hline
 10K & 1.7 & 1.5\\ 
 \hline
 100K & 13 & 9.6 \\
 \hline
 500K & 67 & 46.4 \\
 \hline
 1M & 135 & 99.7\\
 \hline
10M & 1592 & 1144\\ 
\hline
\end{tabular}
\end{center}
\caption{Running times of \textsc{PRC} and one pass.}\label{tab:times}
\end{table}

%
%


\bibliographystyle{plainnat}
\bibliography{refs.bib}

\appendix


\section{Greedy Cutting and Single-linkage.}\label{app:slc-greedy}
Consider a simple algorithm, denoted by \textsc{Greedy Cut}, that picks the interval  with maximum length among $\{[x_i,x_i+1]\}_{i=1:n-1}$ (lets say $[x_m,x_{m+1}]$), and repeats the same operation recursively on $(x_1,\ldots,x_m)$ and $(x_{m+1},\ldots,x_n)$ until the leaves are reached. 
\begin{lemma}
For $d=1$ and under any monotone distance-based similarity measure $\dfun$ \textsc{Greedy Cut} and single-linkage return the same HC tree. Moreover, the edges picked by \textsc{Greedy Cut} are exactly the same edges picked by single-linkage, picked in reverse order. 
\end{lemma}
\begin{proof}
It is known that single-linkage is essentially the Kruskal algorithm (and hence edges picked by single-linkage form a Maximum Spanning Tree (MST)). Clearly, for any monotone distance-based measure the line connecting $x_1$ to $x_n$ is the unique MST (as any tree can be shortcut-ed with this line), and hence single-linkage picks intervals $\{[x_i,x_{i+1}]\}_{i=1:n-1}$ in increasing order of their lengths. At the same time, \textsc{Greedy Cut} picks also the same intervals, but in decreasing order of their length. Moreover, because single-linkage merges the edges picked by \textsc{Greedy Cut} in the reverse ordering (and it creates the HC tree from bottom to top), it returns the same HC tree as \textsc{Greedy Cut}.
\end{proof}

\begin{remark} As a simple observation, \textsc{Greedy Cut} is equivalent to \emph{reverse-Kruskal} in 1D; It starts from all edges,  goes over them in increasing order of weights (here, in decreasing order of lengths), and only keeps an edge when its removal makes the graph disconnected. We claim the first edge picked by reverse-kurskal corresponds to the interval $[x_m,x_{m+1}]$ with the maximum length, and hence the equivalence between the two algorithms by induction. To prove the claim, the line between $x_1$ and $x_n$ keeps the graph connected, so the first picked edge corresponds to an interval $(x_i,x_{i+1})$. Also, it should be the interval $[x_m,x_{m+1}]$ with maximum length. Moreover, removing this edge makes the graph disconnected, because otherwise there is was cross edge between the left-side $(x_1,...,x_m)$ and the right-side $(x_{m+1},\ldots,x_n)$. The length of such an edge is more than the length of $(x_m,x_{m+1})$, so it should have been removed before, a contradiction.
\end{remark}

We finish the section by demonstrating the lack of performance of single-linkage (and hence \textsc{Greedy Cut}) through an example, which justifies using both average-linkage and \textsc{Random Cut} as \emph{smooth} versions of single-linkage for 1D data points.
\begin{lemma} For $d=1$, single-linkage can obtain at most $\tfrac{1}{2}$ of the optimum objective, and $\tfrac{1}{2}$ of the objective values of average-linkage and \textsc{Random Cut}, under Gaussian kernels.
\end{lemma}
\begin{proof}
Consider an example with $n$ equally spaced points on a line, where the distance between any two adjacent node is $\Delta$. We now slightly move the points so that 
weight of $(x_i,x_{i+1})$ is $(\sqrt{2 \pi} \sigma)^{-n} (1-(i-1)\epsilon) e^{- \Delta^2/2\sigma^2}$ for $i=1:n-1$. Now, single-linkage peels off points in the order $x_1,\ldots,x_n$. Roughly speaking, we let $\Delta/\sigma^2$ to be large enough so we can ignore the similarity weights between any two non-adjacent points. Hence single-linkage gets an objective value of $\approx \left(\sum_{i=1}^{n-1} (n-i)\right)(\sqrt{2 \pi} \sigma)^{-n} e^{- \Delta^2/2\sigma^2}$, which evaluates to $\left(n^2/2+o(n^2)\right)(\sqrt{2 \pi} \sigma)^{-n} e^{- \Delta^2/2\sigma^2}$. Average-linkage returns the symmetric binary tree (and hence the objective value is $$\left((n-2)(n-1)-n\log(n)\right)(\sqrt{2 \pi} \sigma)^{-n} e^{- \Delta^2/2\sigma^2}=\left(n^2+o(n^2)\right)(\sqrt{2 \pi} \sigma)^{-n} e^{- \Delta^2/2\sigma^2}$$. Random returns a random binary tree, with (roughly speaking) a similar objective value of $$\left(n^2+o(n^2)\right)(\sqrt{2 \pi} \sigma)^{-n} e^{- \Delta^2/2\sigma^2}$$ in expectation. The fact that optimum objective value is as large as these two quantities finishes the proof.
\end{proof}

\section{Deferred Discussions in \texorpdfstring{\Cref{sec:experiment}}{}}
\label{sec:appendix-experiment}
\paragraph{Datasets.}
We use real data in our experiments. To be able to measure the performance of produced Hierarchical Clusterings, we need to compute \mbench, and hence we use a small dataset. To be able to show linear scaling of running time even for tens of millions of datapoints, we will use large datasets of sizes 10K to 10M.

\emph{The small dataset:} The Zoo dataset contains 100 animals given as 16-dimensional vectors, forming 7 different classes (\eg mammals, amphibians, etc.). The features contain information about the animal's characteristics (\eg if it has tail, hair, fins, etc.). Here, we want to showcase the quality of the solution produced by \textsc{PRC}; in the case of the small dataset, we can afford keeping track of two benchmarks: the performance of the widely-used \textsc{Spectral} algorithm and the \mbench \footnote{$\mbench$ is based on triples of datapoints, thus taking $O(n^3)$ time.}. 

\emph{The large dataset:} In the SIFT10M dataset, each data point is a SIFT feature, extracted from Caltech-256~\cite{griffin2007caltech} by the open source VLFeat library~\cite{vedaldi2010vlfeat}.
Caltech-256 is used as a computer vision benchmark image data set, that features a total of 256 different classes with high intra-class variations for each category. Each datapoint is a 128-dimensional vector and similar to the Zoo dataset, we use this information as features to perform fast Hierarchical Clustering. Here, we want to test the scalability of our algorithm so we use successively larger datasets from the SIFT10M dataset of 10K, 100K, 500K, 1M and finally 10M datapoints by truncating the data file as necessary.

\end{document}